\documentclass[10pt,journal,compsoc]{IEEEtran}

\usepackage{amsmath,amsthm,amssymb}
\usepackage{mathtools}
\usepackage{fontenc}
\usepackage{enumerate}
\usepackage{bm,bbm}
\usepackage{graphicx,subcaption} \graphicspath{ {figures/} }
\usepackage{color}

\ifCLASSOPTIONcompsoc
\usepackage[nocompress]{cite}
\else
\usepackage{cite}
\fi

\newtheorem{theorem}{Theorem}[section]

\newtheorem{proposition}[theorem]{Proposition}
\newtheorem{corollary}[theorem]{Corollary}
\newtheorem{definition}[theorem]{Definition}

\newtheorem{observation}[theorem]{Observation}

\newtheorem*{problemstatement*}{Registration Problem}

\DeclarePairedDelimiter \norm {\lVert}{\rVert}
\def \x  {\textbf{x}}
\def \y  {\mathbf{y}}
\def \bX {\mathbf{X}}
\def \bY {\mathbf{Y}}
\def \cE {\mathcal{E}}
\def \cI {\mathcal{I}}
\def \cP {\mathcal{P}}
\def \cQ {\mathcal{Q}}
\def \cR {\mathcal{R}}
\def \cS {\mathcal{S}}
\def \cT {\mathcal{T}}
\def \cU {\mathcal{U}}
\def \R  {\mathbb{R}}

\makeatletter
\newcommand{\leqnomode}{\tagsleft@true\let\veqno\@@leqno}
\newcommand{\reqnomode}{\tagsleft@false\let\veqno\@@eqno}
\makeatother

\begin{document}
\title{On Uniquely Registrable Networks}
\author{Aditya~V.~Singh,~\IEEEmembership{Student~Member,~IEEE}
	and~Kunal~N.~Chaudhury,~\IEEEmembership{Senior~Member,~IEEE}
	
\thanks{The authors are with the Department of Electrical Engineering, Indian Institute of Science, Bangalore, India. E-mail: adityavs@iisc.ac.in, kunal@iisc.ac.in. Corresponding author: Aditya~V.~Singh.}}
%

%
%

\IEEEtitleabstractindextext{	
\begin{abstract}
Consider a network with $N$ nodes in $d$-dimensional Euclidean space, and $M$ subsets of these nodes $P_1,\cdots,P_M$. Assume that the nodes in a given $P_i$ are observed in a local coordinate system. The \emph{registration} problem is to compute the coordinates of the $N$ nodes in a global coordinate system, given the information about $P_1,\cdots,P_M$ and the corresponding local coordinates. The network is said to be \emph{uniquely registrable} if the global coordinates can be computed uniquely (modulo Euclidean transforms). We formulate a necessary and sufficient condition for a network to be uniquely registrable in terms of rigidity of the \emph{body graph} of the network. A particularly simple characterization of unique registrability is obtained for planar networks. Further, we show that $k$-vertex-connectivity of the body graph is equivalent to quasi $k$-connectivity of the bipartite \emph{correspondence graph} of the network. Along with results from rigidity theory, this helps us resolve a recent conjecture due to Sanyal et al. (IEEE TSP, 2017) that quasi $3$-connectivity of the correspondence graph is both necessary and sufficient for unique registrability in two dimensions. We present counterexamples demonstrating that while quasi $(d+1)$-connectivity is necessary for unique registrability in any dimension, it fails to be sufficient in three and higher dimensions.
\end{abstract}

}

\maketitle

\IEEEraisesectionheading{\section{Introduction}}
\IEEEPARstart{W}{e} consider the problem of registering nodes of a network in a global coordinate system, given the coordinates of overlapping subsets of nodes in different local coordinate systems. Registration problems of this kind arise in situations where we wish to reconstruct an underlying global structure from multiple local sub-structures, such as in sensor network localization, multiview registration, protein structure determination, and manifold learning \cite{sanyal,mihai12,gortler2013,krishnan,sharp2004,mihaimol12,fang2013,zhang2004}. 
For instance, consider an adhoc wireless network consisting of geographically distributed sensor nodes with limited radio range. To make sense of the data collected from the sensors, one usually requires the positions of the individual sensors. The positions can be found simply by attaching a GPS with each sensor, but this is often not feasible due to cost, power, and weight considerations. On the other hand, we can estimate (using time-of-arrival) the distances between sensor that are within the radio range of each other \cite{mao2007}. The problem of estimating sensor locations from the available inter-sensor distances is referred to as sensor network localization (SNL) \cite{mao2007,shang2004}. 
Efficient methods for accurately localizing small-to-moderate sized networks have been proposed over the years \cite{soares15,simonetto14,wang08,biswas06}. However, these methods typically cannot be used to localize large networks. To address this, scalable divide-and-conquer approaches for SNL have been proposed in \cite{arap10,mihai12,knc2015,sanyal}, where the large network is first subdivided into smaller subnetworks which can be efficiently and accurately localized (pictured in Fig. \ref{fig:regscen}(a)). 
Each subnetwork (called patch) is then localized independent of other subnetworks. Thus, the coordinates returned for a patch will in general be an arbitrarily rotated, flipped, and translated version of the ground-truth coordinates (Fig. \ref{fig:regscen}(b)). The network is thus divided into multiple patches, where each patch can be regarded as constituting a local coordinate system which is related to the global coordinate system by an unknown rigid transform. We now want to assign coordinates to all the nodes in a global coordinate system based on these patch-specific local coordinates. 

The registration problem also comes up in multiview registration, where the objective is to reconstruct a $3$D model of an object based on partial overlapping scans of the object (Fig. \ref{fig:regappl}(a),(b)). Here, the scans can be seen as patches, which are to be registered in a global reference frame via rotations and translations. Similar situation arises in protein conformation (Fig. \ref{fig:regappl}(c),(d)), where we are required to determine the $3$D structure of a protein (or other macromolecule) from overlapping fragments \cite{mihaimol12,fang2013}.

In such problems, a question that naturally arises is that of uniqueness: Can we uniquely identify the global topology of the network that is consistent with the information in the various local coordinate systems? Additionally, do we have computationally efficient tests to determine if the network is uniquely registrable? In this paper, we investigate these questions using results from graph rigidity theory.

\begin{figure*}[h]
	\centering
	\includegraphics[width=\linewidth]{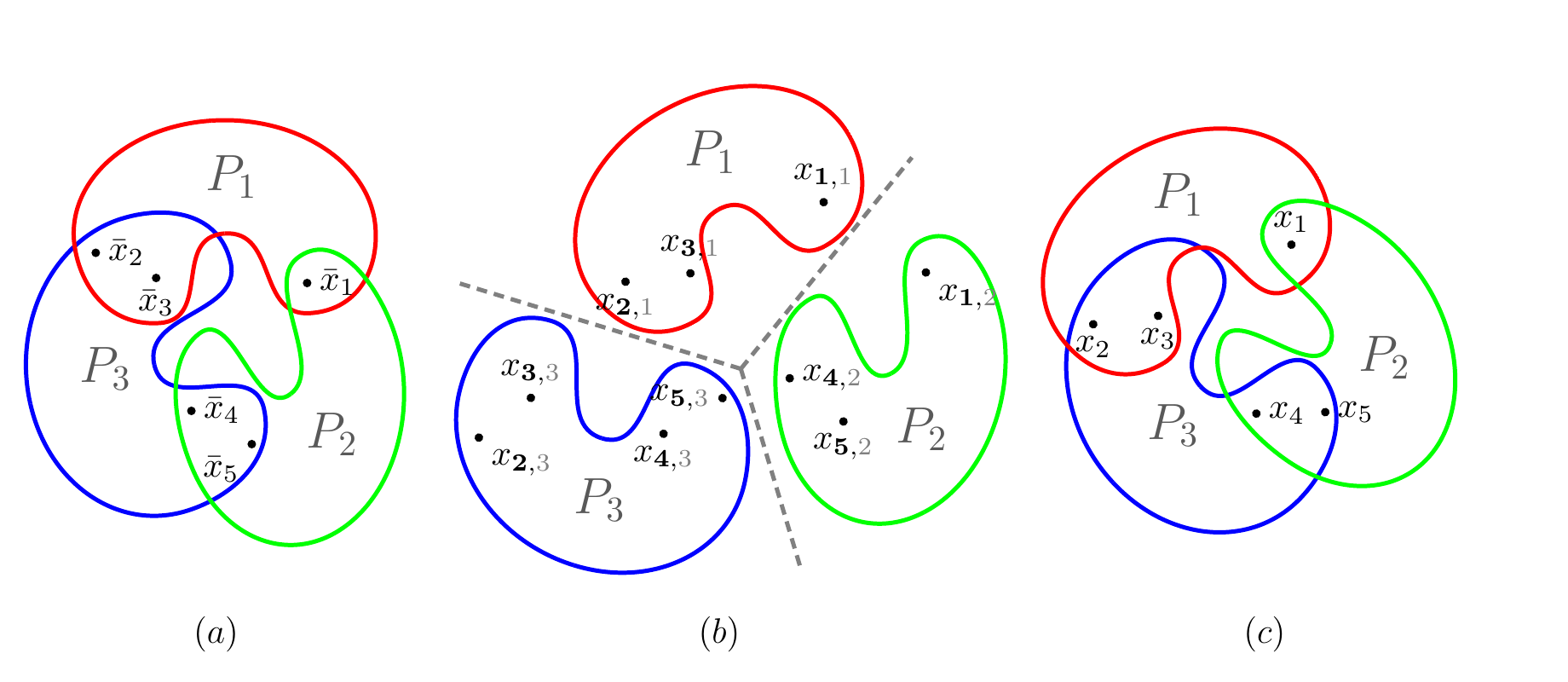}
	\caption
	{
		Typical registration scenario. $(a)$ Ground truth network; $P_1$, $P_2$, $P_3$ are the subnetworks (patches), $(b)$ Three local coordinate systems, with $x_{k,i}$ denoting the coordinate of the $k$-th node in the $i$-th local coordinate system (based on this information, we would like to recover the ground truth network), $(c)$ Reconstructed network. Note that the reconstructed network and the ground truth network are related by a global Euclidean transform, which is the best we can do with the given information. If we want to recover the ground truth network exactly, we need to incorporate at least $d+1$ \emph{anchor nodes} in our network, which are the nodes in the network whose global coordinates are known \emph{a priori}. The anchor nodes (if any) can be considered as forming a patch of their own \cite{sanyal}, and thus our analysis incurs no loss in generality by ignoring their presence.
	}
	\label{fig:regscen}
\end{figure*}

\begin{figure}
	\centering
	\begin{subfigure}[b]{0.49\linewidth}
		\centering
		\includegraphics[width=\linewidth]{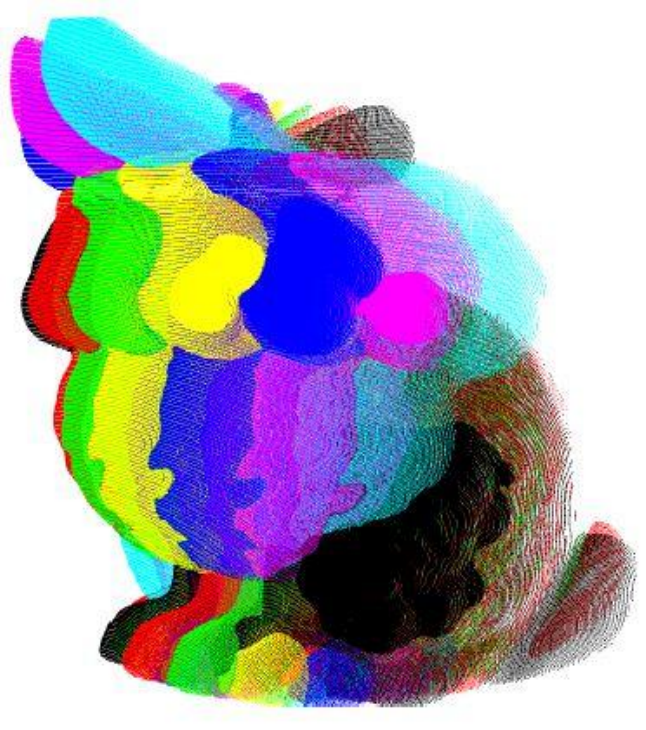}
		\caption
		{
			Partial $3$D scans.		}
		\label{fig:bunny1}
	\end{subfigure}
	\begin{subfigure}[b]{0.49\linewidth}
		\centering
		\includegraphics[width=\linewidth]{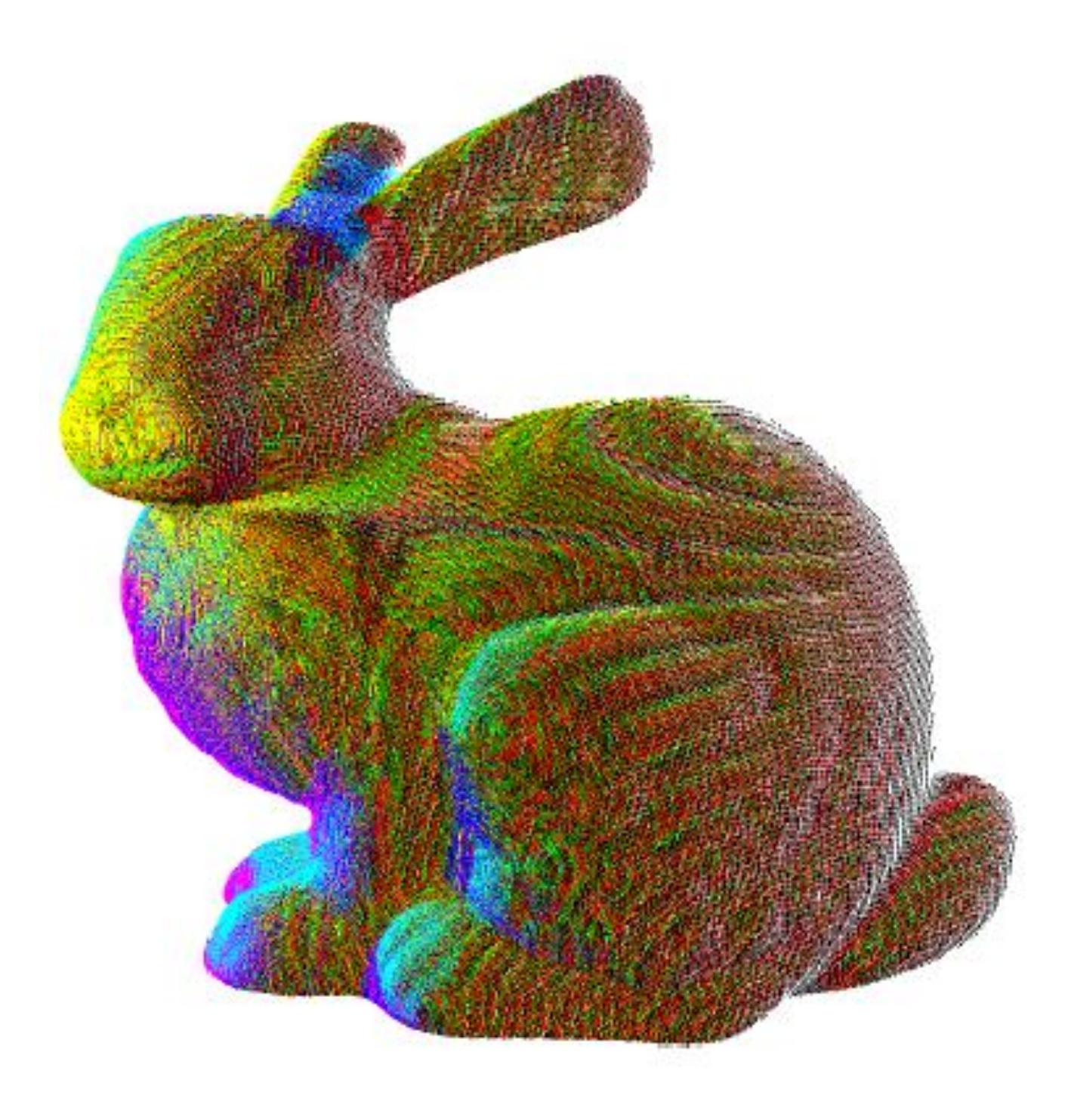}
		\caption
		{
			After registration.
		}
		\label{fig:bunny2}
	\end{subfigure}	
	\begin{subfigure}[b]{0.49\linewidth}
		\centering
		\includegraphics[width=\linewidth]{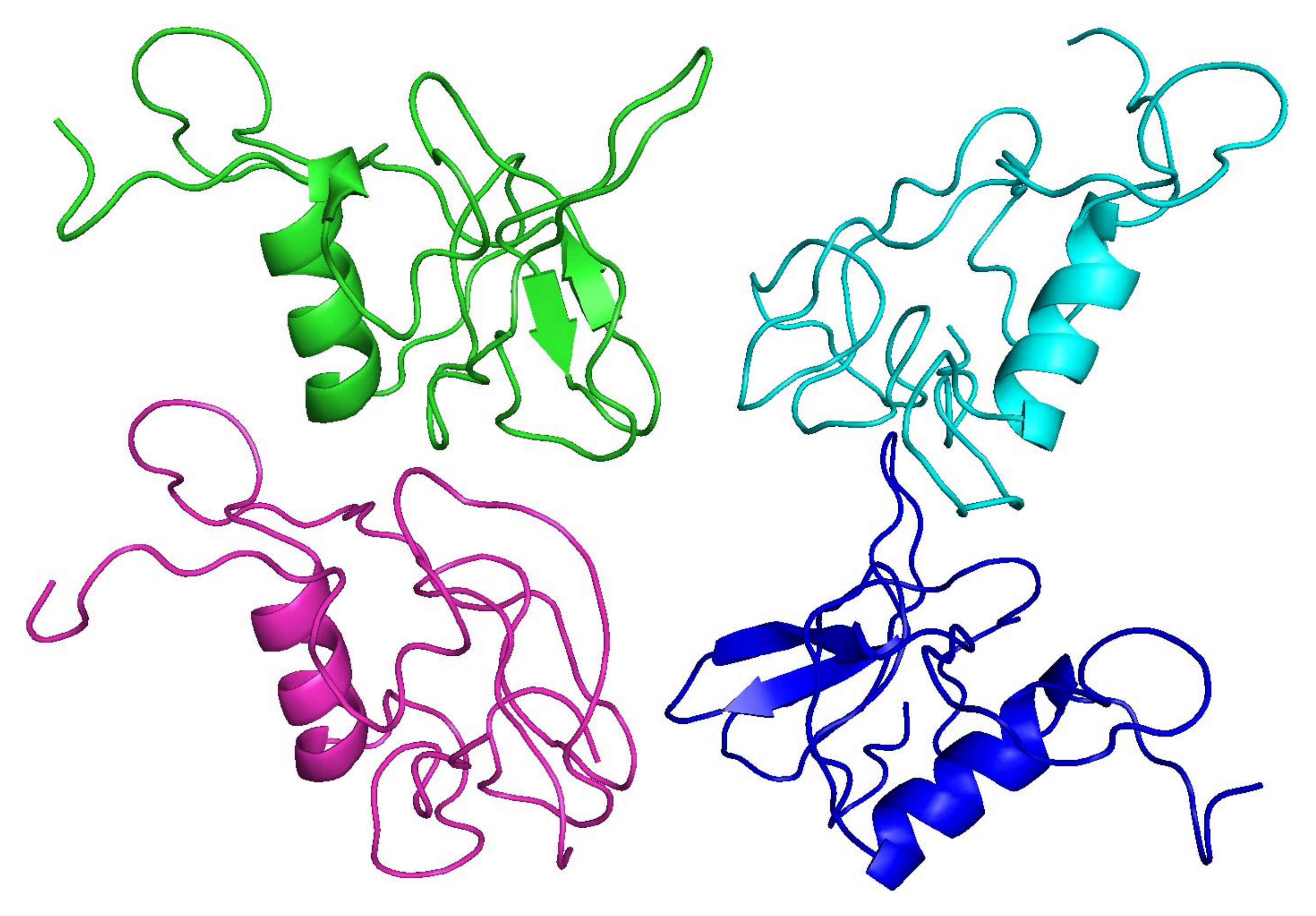}
		\caption
		{
			Fragments of a protein.
		}
		\label{fig:protein1}
	\end{subfigure}
	\begin{subfigure}[b]{0.49\linewidth}
		\centering
		\includegraphics[width=\linewidth]{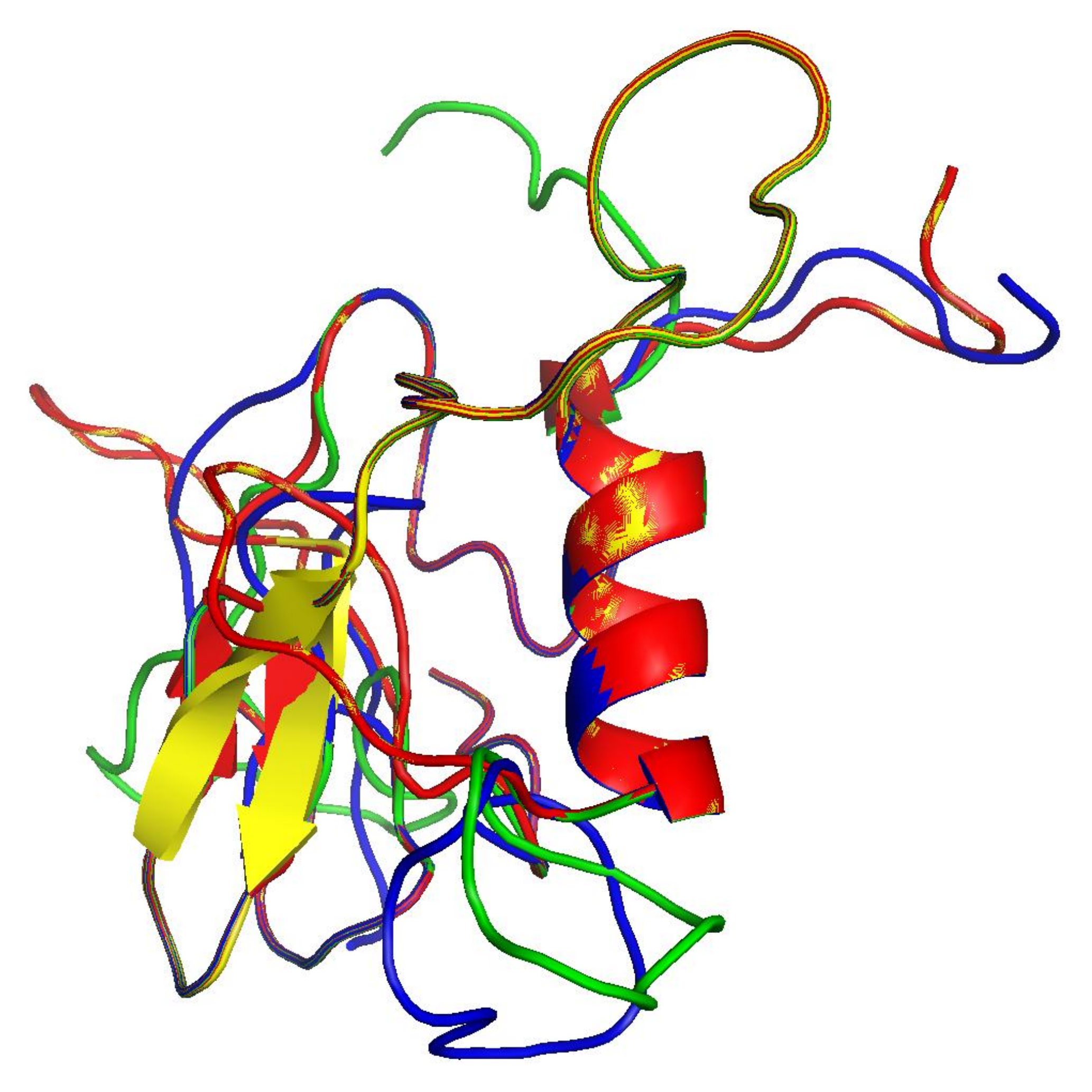}
		\caption
		{
			After registration.
		}
		\label{fig:protein2}
	\end{subfigure}
	\caption
	{ 
		Registration in action. (a),(b): Registration of multiview scans \cite{miraj}.  (c), (d): Registration of protein fragments.
	}
	\label{fig:regappl}
\end{figure}

\subsection{Problem Formulation}
To better facilitate discussion of our contribution, and how it fits in the context of previous work in this area, we formally describe the registration problem, and discuss the notion of uniqueness. Suppose a network consists of $N$ nodes in $\R^d$, which we label using\footnote{we use $[m:n]$ to denote the set of integers $\{m,\ldots,n\}$.} $\cS = [1 : N]$. Let $P_1,\cdots,P_M$ be subsets of $\cS$. We refer to each $P_i$ as a \textit{patch} and let  $\cP = \{P_1,\cdots,P_M\}$ be the collection of patches. A natural way to represent the node-patch correspondence is using the bipartite graph $\Gamma_C = (\cS,\cP,\cE)$, where $(k,i) \in \cE$ if and only if node $k$ belongs to patch $P_i$. We refer to $\Gamma_C$ as the \emph{correspondence graph}.
Let $\bar{x}_1,\ldots,\bar{x}_N \in  \R^d$ be the true coordinates of the $N$ nodes in some global coordinate system. We associate with each patch a local coordinate system: If $(k,i) \in \cE$, let $x_{k,i} \in \R^d$ be the local coordinates of node $k$ in patch $P_i$.  In other words, if $\bar{\cR}_i$ is the Euclidean transform (defined with respect to the global coordinate system) associated with patch $P_i$, then 
\begin{equation} \label{groundtruth}
\bar{x}_k = \bar{\cR}_i(x_{k,i}), \qquad (k,i) \in \cE.
\end{equation} 
We will refer to $\bar{\cR}_i$ as the \emph{patch transform} associated with patch $P_i$. We are now ready to give a precise statement of the registration problem.

\begin{problemstatement*}
	Given a correspondence graph $\Gamma_C = (\cS,\cP,\cE) $ and local coordinates $\{ x_{k,i}: \ (k,i) \in \cE \}$, find  $\bX = (x_k)_{k=1}^N$, and $\bm{\cR} = (\cR_i)_{i=1}^M$, such that for $(k,i) \in \cE$,
	\begin{equation}
	x_k = \cR_i(x_{k,i}).
	\label{reg}
	\tag{REG}
	\end{equation} 
\end{problemstatement*}

Clearly, the true global coordinates $(\bar{x}_k)_{k=1}^N$ and the patch transforms $(\bar{\cR}_i)_{i=1}^M$ satisfy \ref{reg}. But is this solution unique? This is a fundamental question one would be faced with when coming up with algorithmic solutions to the registration problem \cite{knc,sanyal}. Of course, by uniqueness, we mean uniqueness up to  \emph{congruence}, i.e., any two solutions that are related through a Euclidean transform are considered identical. Note that a solution to \ref{reg} has two components: the global coordinates, and the patch transforms. We will define uniqueness for each of these components. Suppose $(\bX, \bm{\cR})$ is a solution to \ref{reg}. By \emph{uniqueness of global coordinates}, we mean that given any other solution $(\bY, \bm{\cT})$ to \ref{reg}, there exists a Euclidean transform $\cQ$ such that $y_k = \cQ(x_k), k \in \cS$. Similarly, by \emph{uniqueness of patch transforms}, we mean that there exists a Euclidean transform $\cU$ such that  $\ \cT_i = \cU \circ \cR_i, i \in [1 : M]$, where $\circ$ denotes the composition of transforms. At this point, we make the following observation.

\begin{observation} \label{obs1}
	It is clear that uniqueness of patch transforms implies uniqueness of global coordinates. That is, given two solutions $(\bX,\bm{\cR})$ and $(\bY,\bm{\cT})$ to \ref{reg}, if there exists a Euclidean transform $\cU$, such that $\cT_i = \cU \circ \cR_i, i \in [1 : M]$, then there exists a Euclidean transform $\cQ$, such that $y_k = \cQ(x_k), k \in \cS$ (in particular, take $\cQ = \cU$). However, uniqueness of global coordinates does not imply uniqueness of patch transforms. That is, given two solutions $(\bX,\bm{\cR})$ and $(\bY,\bm{\cT})$ to \ref{reg}, there may not exist a Euclidean transform $\cU$, such that $\cT_i = \cU \circ \cR_i,  i \in [1 : M]$, even if there exists a Euclidean transform $\cQ$, such that $y_k = \cQ(x_k), k \in \cS$. (This is explained with an example in Fig. \ref{fig:k3patch}.)
\end{observation}

\begin{figure}[h]
	\centering
	\includegraphics[width=0.7\linewidth]{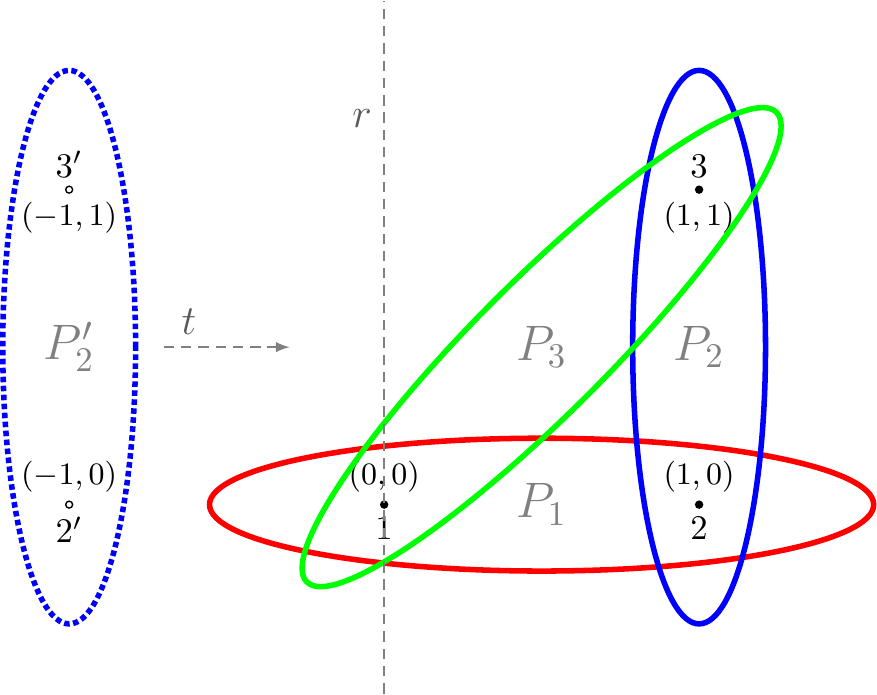}
	\caption
	{
		Consider the nodes $\cS = \{1,2,3\}$, and the patches $\cP = \{P_1,P_2,P_3\}$, where $P_1=\{1,2\}, P_2=\{2,3\}, P_3=\{1,3\}$. The true global coordinates are $\bar{\bX} = ((0,0),(1,0),(1,1))$, and the true patch transforms are $\bar{\bm{\cR}} = (\cI_d,\cI_d,\cI_d)$, where $\cI_d$ is the identity transform (i.e., each patch coordinate system is same as the global coordinate system). Consider the Euclidean transform $\mathcal{T}$, which is a reflection along the dotted line marked $r$, followed by a translation of 2 units along the dotted ray marked $t$. Let $\bm{\cR} = (\cI_d,\mathcal{T},\cI_d)$. Notice that even though both $(\bar{\bX},\bar{\bm{\cR}})$ and $(\bar{\bX},\bm{\cR})$ are solutions to \ref{reg}, $\bm{\cR}$ is not congruent to $\bar{\bm{\cR}}$.
	}
	\label{fig:k3patch}
\end{figure}

Notice that each patch has just two nodes in the example in Fig. \ref{fig:k3patch}. However, we know that a Euclidean transform in $\R^d$ is completely specified by its action on a set of $d+1$ non-degenerate nodes\footnote{A set of nodes in $\R^d$ is said to be \emph{non-degenerate} if their affine span is $\R^d$.}. Equivalently, if  $d+1$ or more non-degenerate nodes are left fixed by a Euclidean transform, then the transform must be identity. This leads to the following proposition.

\begin{proposition} \label{prop:nondegenunique}
	If every patch contains at least $d+1$ non-degenerate nodes, then uniqueness of global coordinates is equivalent to uniqueness of patch transforms.
\end{proposition}

\begin{proof}
	In Observation \ref{obs1}, we saw that uniqueness of patch transforms implies uniqueness of global coordinates. Thus, we need only prove the converse: that uniqueness of global coordinates implies uniqueness of patch transforms.
	Suppose we have two solutions $(\bX,\bm{\cR})$ and $(\bY,\bm{\cT})$. 
	Following the uniqueness of global coordinates, there exists a Euclidean transform $\cQ$, such that $y_k = \cQ(x_k), k \in \cS$. 
	Fix some $i \in [1 : M]$. 
	Since $(\bY,\bm{\cT})$ is a solution to \ref{reg}, we have $y_k = \cT_i(x_{k,i}), k \in P_i$. 
	Thus, $\cQ(x_k) = \cT_i(x_{k,i})$, or $x_k = (\cQ^{-1}\circ \cT_i)(x_{k,i}), k \in P_i$. 
	On the other hand, since $(\bX,\bm{\cR})$ is also a solution to \ref{reg}, we have $x_k = \cR_i(x_{k,i}), k \in P_i$. 
	Combining the above, we get $(\cQ^{-1}\circ \cT_i)(x_{k,i})=  \cR_i(x_{k,i}), k \in P_i$. Since $|P_i| \geq d+1$, 
	it follows that $\cQ^{-1}\circ\cT_i = \cR_i$, or $\cT_i = \cQ\circ \cR_i$. 
	This holds for every $i \in [1 : M]$, which proves our claim.
\end{proof}

In other words, if every patch contains at least $d+1$ non-degenerate nodes, we need not distinguish between \emph{uniqueness of global coordinates} and \emph{uniqueness of patch transforms}, and we can generally talk about \emph{unique registrability} (i.e. uniqueness of solution to \ref{reg}) without any ambiguity. Intuitively, it is clear that for \ref{reg} to have a unique solution, there must be sufficient overlap among patches. In particular, $\Gamma_C$ must be connected. In Section \ref{sec:back2prob}, we will see that the notion of uniqueness of a solution to \ref{reg} is essentially combinatorial in nature for \emph{almost every} instance of the problem.

\subsection{Related Work}
The correspondence graph $\Gamma_C = (\cS,\cP,\cE)$ encodes the pattern of overlap among patches, which makes it desirable to relate the problem of unique registrability to the properties of $\Gamma_C$. In \cite{knc}, the authors propose a lateration criterion which guarantees unique registrability. We recall that $\Gamma_C$ is said to be \emph{laterated} if there exists a reordering of the patch indices such that $P_1$ contains at least $d+1$ non-degenerate nodes, and $P_i$ and $P_1 \cup P_2 \cup \cdots \cup P_{i-1}$ have at least $d+1$ non-degenerate nodes in common for $i \geq 2$. This criterion, however, has two major shortcomings. First, an efficient test for lateration is not known. Second, lateration is a rather strong condition. For instance, see Fig. \ref{fig:bodygraph}, where $\Gamma_C$ is not laterated, but, as we will see later, the network is uniquely registrable. More recently, the notion of \emph{quasi connectedness} of $\Gamma_C$ was introduced in \cite{sanyal}, which was shown to be necessary for unique registrability, and conjectured to be sufficient.

In a related work \cite{aspnes2006}, rigidity theory is used to deal with unique localizability of nodes in a general sensor network localization problem, where, given inter-node distances of a subset of node-pairs, a graph is constructed with the vertices corresponding to the nodes, and an edge between every node-pair whose inter-node distance is given; it is demonstrated that this graph has to be globally rigid for unique localizability of the sensor network. In the context of divide-and-conquer approach to molecular reconstruction problem, the authors in \cite{mihaimol12} use results from graph rigidity theory to obtain uniquely localizable patches. Tools from rigidity theory have also been used in network design problem \cite{shames2015rigid}, and in quantifying robustness of networks \cite{eren2015combinatorial}.

\subsection{Contribution and Organization}
Our contribution in this paper is two-fold. First, we bring in the notion of \emph{body graph}, introduced in \cite{gortler2013} in the context of affine rigidity, and show that unique registrability of a network is equivalent to global rigidity of the body graph of the network. This, in effect, opens up the possibility of using standard tools and techniques from rigidity theory to formulate conditions for unique registrability. Second, we address the conjecture posed in \cite{sanyal}, namely that  quasi $(d+1)$-connectivity of $\Gamma_C$ is necessary and sufficient for unique registrability in $\R^d$. We show that quasi connectivity of $\Gamma_C$ is equivalent to vertex-connectivity of the body graph, and then use combinatorial characterizations of rigidity in two dimensions to establish the conjecture for $d=2$. This, in particular, gives a simple characterization of unique registrability for planar networks, where we need only check quasi $3$-connectivity of $\Gamma_C$. Next, we give counterexamples to show that the conjecture is false when $d \geq 3$. 

The rest of the paper is organized as follows. In Section \ref{sec:rigidity}, we review relevant definitions and results from rigidity theory. In Section \ref{sec:back2prob}, introduce the notion of body graph and derive our main results on unique registrability. In Section \ref{sec:quasiconn}, we resolve the conjecture posed in \cite{sanyal}.  We summarize our results in Section \ref{sec:concl}. Detailed proofs of some of the technical results from Sections \ref{sec:back2prob} and \ref{sec:quasiconn} are given in Section \ref{sec:supplementary}.

\subsection{Graph Notations}
We will work with undirected graphs in this paper. If $H$ is a subgraph of $G = (V,E)$, which we denote by $H \subseteq G$, then $V(H)$ denotes the set of vertices of $H$, and $E(H)$ denotes the set of edges of $H$. A complete graph (or clique) on $n$ vertices is denoted by $K_n$. Given a graph $G = (V,E)$, and a set $V' \subseteq V$, the subgraph induced by $V'$ is the graph $G' = (V',E')$, where $E' = \{(i,j) \in E : i,j \in V'\}$. The degree of a vertex $v$ of a graph is the number of edges incident on $v$. A path in a graph $G = (V,E)$ is an ordered sequence of distinct vertices $v_1,\cdots,v_n \in V$ such that $(v_i,v_{i+1}) \in E, 1 \leq i \leq n-1$. We denote a path by $v_1-\cdots-v_n$; $v_1$ and $v_n$ are called the end vertices of the path, and every other vertex of the path is  an internal vertex. If $v_1 = a$ and $v_n = b$, we say that the path connects $a$ and $b$, or that $v_1-\cdots-v_n$ is a path between $a$ and $b$. Given subgraphs $A$ and $B$, an $A$-$B$ path is a path $v_1-\cdots-v_n$ where $v_1 \in V(A)$ and $v_n \in V(B)$. Given a subgraph $A$, a path $v_1-\cdots-v_n$ is said to be within $A$, if $v_i \in V(A)$ for every $i \in [1 \colon N]$. Two paths are said to be disjoint if they do not have any vertex in common. Two paths are said to be independent if they do not have any internal vertex in common. A graph is said to be $k$-connected (or, $k$-vertex-connected) if it has more than $k$ vertices and the subgraph obtained after removing fewer than $k$ vertices remains connected; equivalently, by Menger's theorem \cite{diestel2000}, there exists $k$ independent paths between every pair of vertices of the graph.

\section{Rigidity Theory} \label{sec:rigidity} 
Before moving on to our results, we recall some definitions and results from rigidity theory \cite{asimow1978,asimow1979,connelly2005,gortler2010,jackson2005}.

\subsection{Basic Terminology}	
Given a graph $G = (V,E)$, a $d$-dimensional \emph{configuration} is a map $\mathbf{p}:V \rightarrow \R^d$. The pair $(G,\mathbf{p})$ is called a $d$-dimensional \emph{framework}. 	 
Throughout this paper, $\norm{\cdot}$ denotes the Euclidean norm.

\begin{definition}[Equivalent frameworks] \label{def:eqframe}
	Two frameworks $(G,\mathbf{p})$ and $(G,\mathbf{q})$ are said to be equivalent, denoted by $(G,\mathbf{p}) \sim (G,\mathbf{q})$, if $\norm{\mathbf{p}(u) - \mathbf{p}(v)} = \norm{\mathbf{q}(u) - \mathbf{q}(v)}$,  for every $(u,v) \in E$. 
\end{definition}	

\begin{definition}[Congruent frameworks] \label{def:congframe}
	Two frameworks $(G,\mathbf{p})$ and $(G,\mathbf{q})$ are said to be congruent, denoted by $(G,\mathbf{p}) \equiv (G,\mathbf{q})$, if $\norm{\mathbf{p}(u) - \mathbf{p}(v)} = \norm{\mathbf{q}(u) - \mathbf{q}(v)}$ for every $u,v \in V$.	
\end{definition}

In other words, congruent frameworks are related through a Euclidean transform. Clearly, congruence implies equivalence, but the converse is generally not true (see Fig. \ref{fig:equivcong}).

\begin{figure}[h]
	\centering
	\includegraphics[width=1\linewidth]{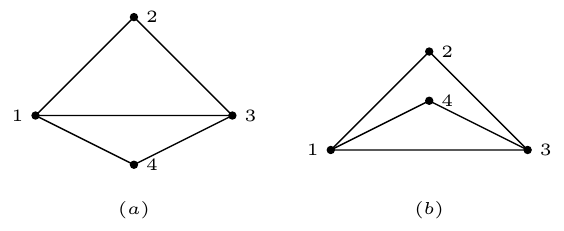}
	\caption
	{Frameworks in $(a)$ and $(b)$ are equivalent because the corresponding edge lengths are equal; however, they are not congruent because the distance between vertices $2$ and $4$ is not equal in the two frameworks. Thus, the framework in $(a)$ is not globally rigid in $\R^2$. On the other hand, it can be shown that the framework is locally rigid in $\R^2$. Observe that there exists no continuous motion in $\R^2$ that takes $(a)$ to $(b)$. Also note that framework $(a)$ is not locally rigid in $\R^3$ since the lower triangle $4$-$1$-$3$ can be rotated in  $3$-dimensional space about the line $1$-$3$ to get framework $(b)$, which is equivalent but non-congruent to framework (a).
	}
	\label{fig:equivcong}
\end{figure}

\begin{definition}[Globally rigidity] 
	A framework $(G,\mathbf{p})$ is said to be globally rigid if any framework equivalent to $(G,\mathbf{p})$ is also congruent to $(G,\mathbf{p})$. 
\end{definition}

This means that given any framework equivalent to a globally rigid framework, there exists a Euclidean transform that relates the two frameworks.

\begin{definition}[Locally rigidity]	
	A framework $(G,\mathbf{p})$ is said to be locally rigid if there exists $\epsilon > 0$ such that any $(G,\mathbf{q}) \sim (G,\mathbf{p})$ satisfying $\norm{\mathbf{p}(v) - \mathbf{q}(v)} \leq \epsilon, v \in V$, is congruent to $(G,\mathbf{p})$. 
\end{definition}

That is, a locally rigid framework cannot be continuously deformed into an equivalent framework (see Fig. \ref{fig:equivcong}).

\subsection{Rigidity and Genericity}	
A fundamental problem in rigidity theory is the following:	
\emph{Given a $d$-dimensional framework $(G,\mathbf{p})$, decide whether it is (locally or globally) rigid in} $\R^d$. In general, the notions of local and global rigidity depend not only on the graph, but also on the configuration (see Fig. \ref{fig:genericity}). This makes testing of rigidity computationally intractable \cite{saxe1980,abbott2008}. A standard way of getting around this is to make an additional assumption of \emph{genericity}. 
A framework (or configuration) is said to be \emph{generic} if there are no algebraic dependencies among the coordinates of the configuration, i.e., the coordinates of the configuration do not satisfy any non-trivial algebraic equation with rational coefficients. For a given graph, the set of non-generic configurations is a measure-zero set in the space of all possible configurations \cite{gluck1975}, and hence \emph{almost every} configuration is generic.

\begin{figure}[h]
	\centering
	\includegraphics[width=1\linewidth]{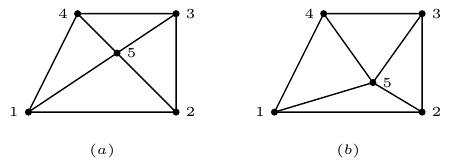}
	\caption
	{Frameworks $(a)$ and $(b)$ with the same underlying graph. Framework $(a)$ is not globally rigid because vertex $4$ can be reflected along the line $1$-$5$-$3$, which results in an equivalent but non-congruent framework. Such an edge-length-preserving reflection is not possible in $(b)$, which is globally rigid.
	}
	\label{fig:genericity}
\end{figure}

We have the following useful proposition which illustrates the utility of the genericity assumption.
\begin{proposition}[\cite{asimow1978,asimow1979,gortler2010}] \label{prop:riggen}
	Local (global) rigidity is a generic property, i.e., either all or none of the generic configurations of a graph form a locally (globally) rigid framework.
\end{proposition}

That is, the assumption of genericity makes local and global rigidity a property of the graph, independent of its configuration. Thus, we can talk of a \emph{graph} being \emph{generically locally (globally) rigid}, by which we mean that every generic configuration of the graph results in a locally (globally) rigid framework. In particular, this opens up the possibility of coming up with combinatorial characterizations for generic local (global) rigidity solely in terms of the graph properties. Combined with the fact that a randomly chosen configuration of a graph is generic with high probability, testing for generic local and global rigidity can be shown to have complexity RP \cite{gortler2010}, which means that there is a polynomial-time randomized algorithm that never outputs a false positive, and outputs a false negative less than half of the time. This fact illustrates the computational tractability afforded by the genericity assumption. We now review some results from rigidity theory relevant to our discussion.

\subsection{Combinatorial Results on Rigidity}
The notion of \emph{redundant rigidity} plays an important role in the context of global rigidity. A graph is said to be redundantly rigid if the graph is generically locally rigid, and remains generically locally rigid after removal of any edge. 
Hendrickson \cite{hendrickson1992} gave the following combinatorial conditions necessary for a graph to be generically globally rigid in $\R^d$. 

\begin{theorem} [\cite{hendrickson1992}] \label{thm:hendrick}
	If a graph $G$ with at least $d+2$ vertices is generically globally rigid in $\R^d$, then
	\begin{enumerate}[(i)]
		\item $G$ is $(d+1)$-connected,
		\item $G$ is redundantly rigid in $\R^d$.
	\end{enumerate}		 
\end{theorem}

Later, Jackson and Jordan \cite{jackson2005} showed that the conditions in Theorem \ref{thm:hendrick} are also sufficient for generic global rigidity in $\R^2$. Thus, we have the following complete combinatorial characterization of generic global rigidity in $\R^2$.

\begin{theorem}[\cite{jackson2005}] \label{thm:jj2d}
	A graph $G$ is generically globally rigid in $\R^2$ if and only if either $G$ is a triangle, or 
	\begin{enumerate}[(i)]
		\item $G$ is $3$-connected, and
		\item $G$ is redundantly rigid in $\R^2$.
	\end{enumerate}
\end{theorem}

Conditions in Theorem \ref{thm:hendrick} are not sufficient for generic global rigidity in $\R^d$ for $d \geq 3$; we shall see instances of such graphs in Section \ref{sec:quasiconn}. 
We now state a result due to \cite{jackson2005,jordan2009} on redundant rigidity in $\R^2$. We do not define the terms `M-circuit' and `M-connected' that appear in the following theorem (as it will take us far afield) and instead refer the reader to \cite{jordan2009} for the definitions. We only need this theorem to derive Proposition \ref{prop:k4redundant}, which we shall use to prove Theorem \ref{thm:body3suff}.

\begin{theorem} [\cite{jordan2009}] \label{thm:jjMconn}
	The following are true in $\R^2$:
	\begin{enumerate}[(i)]
		\item If a graph $G$ is $3$-connected and each edge of $G$ belongs to an M-circuit, then $G$ is M-connected.
		\item If a graph $G$ is M-connected, then $G$ is redundantly rigid.
	\end{enumerate}
\end{theorem}

Theorem \ref{thm:jjMconn}, combined with the fact that complete graph $K_4$ is an M-circuit in $\R^2$ \cite{jackson2005}, leads us to the following proposition.

\begin{proposition} \label{prop:k4redundant}
	If graph $G$ is $3$-connected and each edge belongs to $K_4$, then $G$ is redundantly rigid. 
\end{proposition}

\section{Unique Registrability} \label{sec:back2prob}
In this section, we formulate the necessary and sufficient condition for uniqueness of solution to \ref{reg} (unique registrability). The main result of the section is Theorem \ref{thm:mainthm}, which gives such a condition under the following two assumptions:

\begin{itemize}
	\item [(A1)] Each patch has at least $d+1$ non-degenerate nodes.	
	\item [(A2)] The nodes of the network are in generic positions.
\end{itemize}

We briefly recall the rationale behind the assumptions. Under Assumption (A1), which is grounded in Proposition \ref{prop:nondegenunique}, uniqueness of the global coordinates and uniqueness of the patch transforms become equivalent, making unique registrability a well-defined notion. In practical applications, we can easily force this assumption for divide-and-conquer algorithms \cite{krishnan,knc2015,sanyal}. Assumption (A2), which is grounded in Proposition \ref{prop:riggen}, allows us to formulate conditions for unique registrability for \emph{almost every} problem instance based solely on the combinatorial structure of the problem. 

We now introduce the notion of a body graph, which will help us tie unique registrability to rigidity theory. For a network with correspondence graph $\Gamma_C = (\cS,\cP,\cE)$, consider a graph $\Gamma_B = (V,E)$, where $V = \cS$, and $E = \{(k_1,k_2) : k_1,k_2 \in P_i \text{ for some } i \in [1 : M]\}$. In other words, vertices of $\Gamma_B$ correspond to the nodes in the network, and we connect two vertices by an edge if and only if the corresponding nodes belong to a common patch  (see Fig. \ref{fig:bodygraph}). Observe that subgraph $H_i \subset \Gamma_B$ induced by nodes belonging to patch $P_i$ form a clique. We will call $\Gamma_B$ the body graph of the network. We derive the term body graph from \cite{gortler2013}, where a similar notion was introduced in the context of affine rigidity. 
Using the notion of body graph, we now state our main result, whose proof we defer to Section \ref{subsec:mainproof}.


\begin{figure*}[h]
	\centering
	\includegraphics[width=0.8\linewidth]{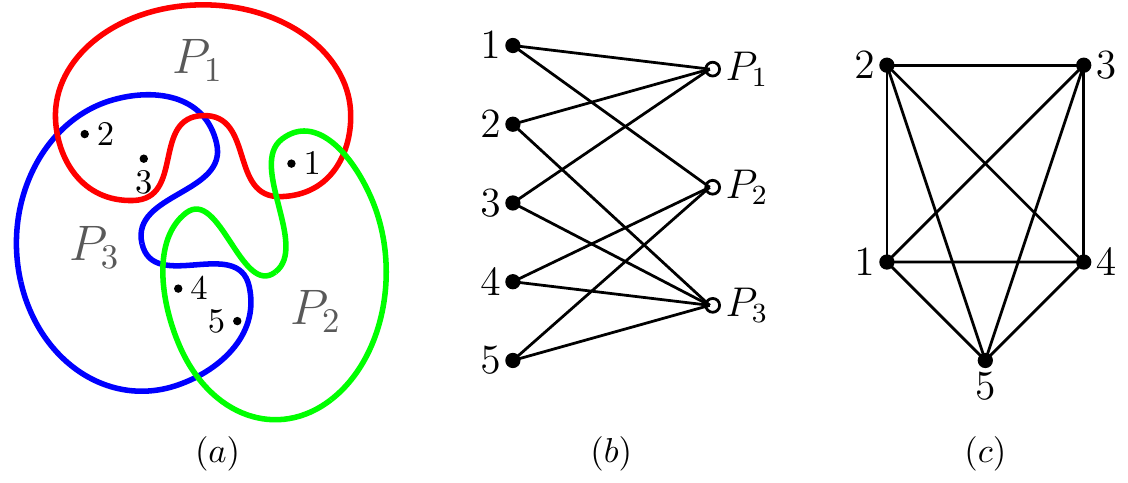}
	\caption
	{For this example, $\cS=[1 : 5]$ and $\cP=\{P_1,P_2,P_3\}$ with $P_1=\{1,2,3\}$, $P_2=\{1,4,5\}$ and $P_3=\{2,3,4,5\}$. $(a)$ Visualization of the node-patch correspondence, $(b)$ Correspondence graph $\Gamma_C=(\cS,\cP,\cE)$, $(c)$ Body graph $\Gamma_B$.
	}
	\label{fig:bodygraph}
\end{figure*}

\begin{theorem} 
	\label{thm:mainthm}
	Under assumptions (A1) and (A2), the ground-truth solution $( \bar{\bX}, \bar{\bm{\cR}} )$ is a unique solution of \ref{reg} if and only if the body graph $\Gamma_B$ is generically globally rigid.
\end{theorem}

The import of Theorem \ref{thm:mainthm} lies in the fact that generic global rigidity in an arbitrary dimension can be tested using a randomized polynomial-time algorithm \cite{gortler2010}.
Moreover, combining Theorem \ref{thm:mainthm} with the combinatorial characterization of generic global rigidity in Theorem \ref{thm:jj2d}, and using additional results from rigidity theory, we get the following characterization of unique registrability for a two-dimensional network, whose proof we defer to Section \ref{subsec:body3suffproof}.

\begin{theorem} \label{thm:body3suff}
	Under assumptions (A1) and (A2), a network is uniquely registrable in $\R^2$ if and only if the body graph $\Gamma_B$ is $3$-connected.
\end{theorem}

The implication of Theorem \ref{thm:body3suff} is that (assuming each patch has at least $3$ nodes) we need only test for $3$-connectivity to establish generic global rigidity of the body graph in $\R^2$. We need not perform an additional check for redundant rigidity, as required by Theorem \ref{thm:jj2d}. As is well-known, $3$-connectivity  can be tested efficiently using linear-time algorithms \cite{jungnickel2008}.

\section{Quasi Connectivity} \label{sec:quasiconn}
In this section, we address the conjecture posed in \cite{sanyal} which asserts that, under Assumption (A1) and the assumption that every set of $d+1$ nodes is non-degenerate, quasi $(d+1)$-connectivity of the correspondence graph $\Gamma_C$ is sufficient for unique registrability in $\R^d$. We prove that, under Assumptions (A1) and (A2), the conjecture holds for $d=2$, but fails to hold for $d \geq 3$. We first recall the definition of quasi connectivity \cite{sanyal}.

\begin{definition} [Quasi $k$-connectivity] \label{quasiconndef}
	The correspondence graph  $\Gamma_C = (\cS,\cP,\cE)$ is said to be quasi $k$-connected if any two vertices in $\cP$ have $k$ or more $\cS$-disjoint paths between them. (A set of paths is $\cS$-disjoint if no two paths have a vertex from $\cS$ in common.)
\end{definition}

\begin{observation} \label{obs:defquasi}
	If the correspondence graph $\Gamma_C$ is quasi $k$-connected, we can infer the following by dint of Definition \ref{quasiconndef}:
	
	\begin{enumerate} [(a)]
		\item There are at least $k$ participating nodes in every patch. (By a participating node, we mean a node that belongs to at least two patches.)
		
		\item Let $\Gamma_B$ be the body graph of $\Gamma_C$. Let $H_i$ be the clique of $\Gamma_B$ induced by patch $P_i$ where $i \in [1 : M]$. Then there are at least $k$ disjoint $H_i$-$H_j$ paths in the body graph, for every $1 \leq i<j \leq M$ (cf. Fig. \ref{fig:counterex1}).
	\end{enumerate}	
\end{observation}

We relate quasi connectivity of the correspondence graph $\Gamma_C$ to connectivity of the associated body graph $\Gamma_B$ in the following theorem, whose proof we defer to Section \ref{subsec:proofquasibodykconn}.

\begin{theorem}[Connectivity of $\Gamma_C$ and $\Gamma_B$] \label{thm:quasibodykconn} \leavevmode
	\begin{enumerate} [(i)]
		\item If the correspondence graph $\Gamma_C$ is quasi $k$-connected, then the body graph $\Gamma_B$ is $k$-connected. 
		
		\item If each patch has at least $k$ nodes and the body graph $\Gamma_B$ is  $k$-connected, then the correspondence graph $\Gamma_C$ is quasi $k$-connected.
	\end{enumerate}
\end{theorem}

We note some corollaries of Theorem \ref{thm:quasibodykconn}. Corollary \ref{cor:quasinecc} was already proved in \cite{sanyal}; we give a short proof using the body graph. Corollary \ref{cor:quasi3suff} establishes the conjecture posed in \cite{sanyal} for $d=2$.

\begin{corollary} \label{cor:quasinecc}
	Under Assumptions (A1) and (A2), quasi $(d+1)$-connectivity of  $\Gamma_C$ is a necessary condition for unique registrability in $\R^d$.
\end{corollary}

\begin{proof}
	From Theorem \ref{thm:mainthm}, unique registrability is equivalent to global rigidity of $\Gamma_B$. From Theorem \ref{thm:hendrick}, $(d+1)$-connectivity of $\Gamma_B$ is a necessary condition for generic global rigidity of $\Gamma_B$ in $\mathbb{R}^d$. The result now follows from Theorem \ref{thm:quasibodykconn}.
\end{proof}

\begin{corollary} \label{cor:quasi3suff}
	Under Assumptions (A1) and (A2), quasi $3$-connectivity of the correspondence graph $\Gamma_C$ is sufficient for unique registrability in $\R^2$.
\end{corollary}

\begin{proof} 
	Follows from Theorem \ref{thm:quasibodykconn} and Theorem \ref{thm:body3suff}.
\end{proof}

Corollary \ref{cor:quasi3suff}, in effect, says that the constraints imposed by quasi $3$-connectivity of $\Gamma_C$ ensure that $\Gamma_B$ is redundantly rigid in addition to being $3$-connected, and hence generically globally rigid in $\R^2$. But this trend does not carry over to $d \geq 3$. We demonstrate it with two examples for $d=3$ (which appear in \cite{jordan2016}), and then note a prescription for generating such counterexamples in higher dimensions. 

\textbf{Example 1.}
\begin{figure}[h]
	\centering
	\includegraphics[width=1\linewidth]{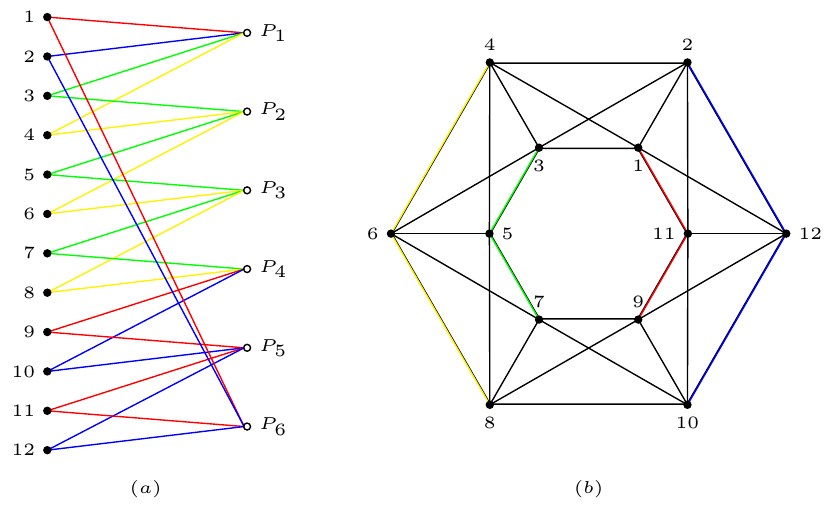}
	\caption{
	The figure shows a counterexample to the sufficiency of quasi $4$-connectivity of the correspondence graph for unique registrability in $\R^3$.
	$(a)$ Correspondence graph $\Gamma_{C1}$, $(b)$ Body graph $\Gamma_{B1}$. The colored paths in $(a)$ show the four $\cS$-disjoint paths between $P_1$ and $P_4$. The corresponding disjoint $H_1$-$H_4$ paths in the body graph $\Gamma_{B1}$ are colored in $(b)$, where $H_1$ and $H_4$ are cliques induced by patches $P_1$ and $P_4$ (see text for details).}
	\label{fig:counterex1}
\end{figure}
Let $\cS = [1 : 12]$, and $\cP = \{P_1,\cdots,P_6\}$. That is, we have $12$ nodes and $6$ patches. Consider the following node-patch correspondence:	
\begin{equation} \label{counterex1}
\begin{aligned}
& P_1 = \{1,2,3,4\}, P_2 = \{3,4,5,6\}, \cdots, \\
& P_5 = \{9,10,11,12\}, P_6 = \{11,12,1,2\}.
\end{aligned}
\end{equation}
The correspondence graph $\Gamma_{C1}$ and the associated body graph $\Gamma_{B1}$ are shown in Fig. \ref{fig:counterex1}. It is easy to verify that $\Gamma_{C1}$ is quasi $4$-connected, or equivalently (Theorem \ref{thm:quasibodykconn}), that $\Gamma_{B1}$ is $4$-connected. But, it can be shown \cite{jordan2016} that the body graph $\Gamma_{B1}$ is minimally rigid in $\R^3$, i.e. $\Gamma_{B1}$ is generically locally rigid, but removing any edge destroys generic local rigidity. Hence $\Gamma_{B1}$ is not redundantly rigid in $\R^3$. This implies, from Theorem \ref{thm:hendrick}, that $\Gamma_{B1}$ is not generically globally rigid, and thus (Theorem \ref{thm:mainthm}), the network is not uniquely registrable in $\R^3$.

\textbf{Example 2.}
\begin{figure}[h]
	\centering
	\includegraphics[width=1\linewidth]{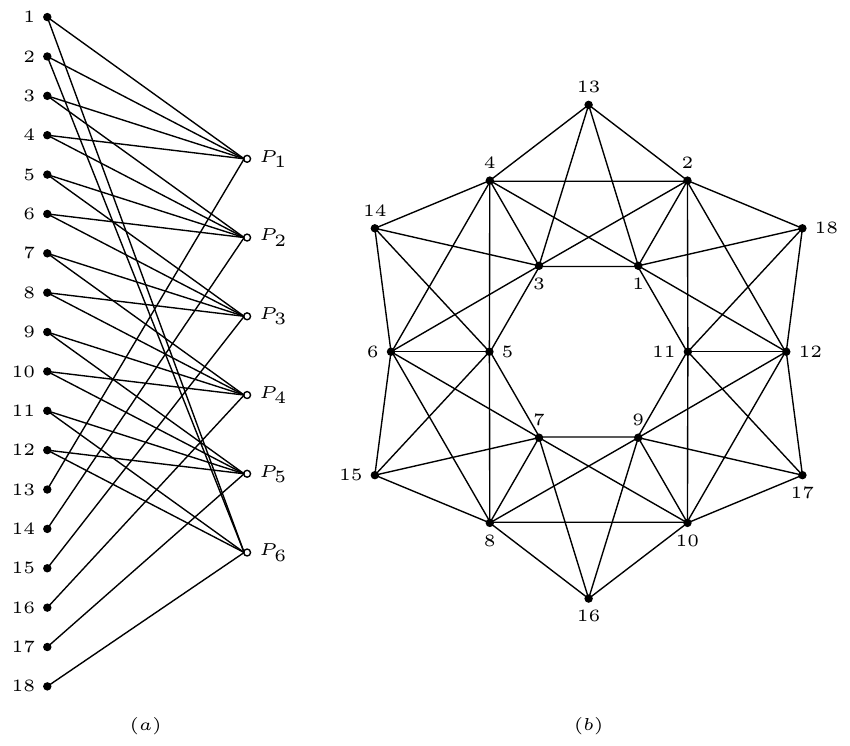}
	\caption{The figure shows a counterexample to sufficiency of quasi $4$-connectivity of the correspondence graph for unique registrability in $\R^3$ even when the body graph is redundantly rigid. $(a)$ Correspondence graph $\Gamma_{C2}$, $(b)$ Body graph $\Gamma_{B2}$ (see text for details).}
	\label{fig:counterex2}
\end{figure}
In this example, we will see that quasi $(d+1)$-connectivity of the correspondence graph is not sufficient for generic global rigidity of the body graph, even when we ensure that the body graph be redundantly rigid. Let $\cS = [1 : 18]$, and $\cP = \{P_1,\cdots,P_6\}$, where
\begin{equation} \label{counterex2}
\begin{aligned}
& P_1 = \{1,2,3,4,\mathit{13}\}, P_2 = \{3,4,5,6,\mathit{14}\}, \cdots, \\
& P_5 = \{9,10,11,12,\mathit{17}\}, P_6 = \{11,12,1,2,\mathit{18}\}.
\end{aligned}
\end{equation}
That is, we have added a non-participating node in each patch of Example 1. The correspondence graph $\Gamma_{C2}$, and the associated body graph $\Gamma_{B2}$ are shown in Fig. \ref{fig:counterex2}. It is easy to verify that $\Gamma_{C2}$ is quasi $4$-connected, or equivalently, that $\Gamma_{B2}$ is $4$-connected. Moreover, $\Gamma_{B2}$ is redundantly rigid \cite{jordan2016}. But, from the fact that $\Gamma_{B1}$ in Example 1 is not generically globally rigid in $\R^3$, it can be deduced (Proposition \ref{prop:addvertexthm}) that $\Gamma_{B2}$ is also not generically globally rigid in $\R^3$. Thus, the network is not uniquely registrable in $\R^3$.

Graphs such as $\Gamma_{B2}$ in Example 2 above, which satisfy both conditions of Theorem \ref{thm:hendrick}, but are not generically globally rigid in $\R^d$, are known as \emph{H-graphs}. By an operation called \emph{coning}, which takes a graph $G$ and adds a new vertex adjacent to every vertex of $G$, a $d$-dimensional H-graph can be turned into a $(d+1)$-dimensional H-graph \cite{jordan2016,frank2011,connelly2010}. In terms of node-patch correspondence, this equates to adding a new node that belongs to every patch. Thus, by applying $d-3$ coning operations to $\Gamma_{B2}$, we can generate a network with a quasi $(d+1)$-connected correspondence graph, which is not uniquely registrable in $\R^d$ for $d>3$.

\section{Discussion}	\label{sec:concl}
In this paper, we looked at the notion of unique registrability of a network through the lens of rigidity theory. Given that there are two families of unknowns inherent in the problem---the global coordinates and the patch transforms---we first addressed the question as to what uniqueness precisely means for the registration problem. We saw that a mild assumption of non-degeneracy makes the notion of uniqueness equivalent for both families of unknowns, which, in turn, makes the notion of unique registrability well-defined. We then introduced the notion of the body graph of a network, which allowed us to reformulate the question of unique registrability into a question about graph rigidity. Specifically, we concluded that unique registrability is equivalent to global rigidity of the body graph. This equivalence opened up the possibility of using non-trivial results from rigidity theory. In particular, we showed that the necessary condition of quasi $(d+1)$-connectivity of the correspondence graph, which was conjectured in \cite{sanyal} to be sufficient for unique registrability in $\R^d$, is indeed sufficient for $d=2$, but fails to be so for $d \geq 3$. The practical utility of these characterizations is that they lead to efficiently testable criteria for unique registrability. In particular, to ascertain unique registrability in $\R^2$, we only need to test quasi $3$-connectivity of the correspondence graph or $3$-connectivity of the body graph (whichever is less expensive). As is well known, three-connectivity can be tested efficiently using linear-time algorithms \cite{jungnickel2008}, whereas, quasi $3$-connectivity can be tested using a variant of existing flow-based algorithms \cite{sanyal}. For $d \geq 3$, unique registrability can be tested simply by testing generic global rigidity of the body graph, for which there exists a polynomial-time randomized algorithm \cite{gortler2010}. The practical utility of these tests is that they can be integrated into existing divide-and-conquer algorithms, including \cite{sanyal}, to ascertain whether the chosen subnetworks can be uniquely registered to localize the entire network.

\section*{Acknowledgements}
The authors thank the editor and the anonymous reviewers for their careful examination of the manuscript and for their useful suggestions; incorporating these suggestions made the presentation more streamlined. K.N. Chaudhury was supported by the DST-SERB Grant SERB/F/6047/2016-2017 from the Department of Science and Technology, Government of India.

\section{Technical Proofs}	\label{sec:supplementary}
In this section, we give proofs for Theorem \ref{thm:mainthm}, Theorem \ref{thm:body3suff} and Theorem \ref{thm:quasibodykconn}.

\subsection{Proof of Theorem \ref{thm:mainthm}} \label{subsec:mainproof}
We show that unique registrability is equivalent to global rigidity of the body graph framework corresponding to the ground-truth. The assumption of genericity (A2) along with Proposition \ref{prop:riggen} (genericity of global rigidity) allows us to remove the dependence on any particular framework, and the theorem is proved. 
We first make some definitions specialized to the registration problem which allow us to express the question of uniqueness registrability in a form amenable to a rigidity theoretic analysis. 

\begin{definition} [Node-patch framework]
	Given a correspondence graph $\Gamma_C = (\cS,\cP,\cE)$, and a map $\textup{\x} : \cS \rightarrow \R^d$ that assigns coordinates to the nodes, the pair $(\Gamma_C,\textup{\x})$ is called a node-patch framework.	
\end{definition}

\begin{definition} [Equivalence of node-patch frameworks] \label{def:ppequiv}
	Two node-patch frameworks $(\Gamma_C,\textup{\x})$ and $(\Gamma_C,\y)$ are said to be equivalent, denoted by $(\Gamma_C,\textup{\x}) \sim (\Gamma_C,\y)$, if $\textup{\x}(k) = \cQ_i  \y(k)$,  $(k,i) \in \cE$, where $\cQ_i$ is a rigid transform.
\end{definition}

\begin{definition} [Congruence of node-patch frameworks] \label{def:ppcong}
	Two node-patch frameworks $(\Gamma_C,\textup{\x})$ and $(\Gamma_C,\y)$ are said to be congruent, denoted by $(\Gamma_C,\textup{\x}) \equiv (\Gamma_C,\y)$, if $\textup{\x}(k) = \cQ  \y(k)$, $\ k \in \cS$, where $\cQ$ is a rigid transform.
\end{definition}

\noindent Given a solution $(\bX,\bm{\cR})$ to \ref{reg}, where $\bX = ( x_k )_{k=1}^N$, $\bm{\cR} = ( \cR_i )_{i=1}^M$, we will denote by $\x$ the map that assigns to node $k$ the coordinate $x_k$, and say that $(\Gamma_C,\x)$ is the node-patch framework corresponding to the solution $(\bX,\bm{\cR})$. 

\begin{proposition} \label{prop:solequiv}
	Let $(\bX,\bm{\cR})$ and $(\bY,\bm{\cT})$ be two solutions to \ref{reg}. Then the corresponding node-patch frameworks $(\Gamma_C,\textup{\x})$ and $(\Gamma_C,\y)$ are equivalent.
\end{proposition}

\begin{proof}
	Since $(\bX,\bm{\cR})$ and $(\bY,\bm{\cT})$ are solutions to \ref{reg}, we have that $\x(k) = \cR_i ( x_{k,i} )$ and $\y(k) = \cT_i ( x_{k,i} )$, $k \in P_i$, $i \in [1:M]$. Thus $\x(k) = \cQ_i  \y(k)$, where $\cQ_i = \cR_i \circ \cT_i^{-1}$.
\end{proof}	

\begin{proposition} \label{prop:equivsol}
	Let $(\bX,\bm{\cR})$ be a solution to \ref{reg} with the corresponding node-patch framework $(\Gamma_C,\textup{\x})$ and let $\y$ be such that $(\Gamma_C,\y) \sim(\Gamma_C,\textup{\x})$.  Then there exists some $\bm{\cT}$ for which $(\bY,\bm{\cT})$ is a solution of \ref{reg}. 
\end{proposition}

\begin{proof}
	Indeed, $(\Gamma_C,\y) \sim (\Gamma_C,\x)$ implies that there exists rigid transforms $(\cQ_i )_{i=1}^M$ such that $\y(k) = \cQ_i  \x(k)$,  $(k,i) \in \cE$. Since $(\bX,\cR)$ is a solution to \ref{reg}, we have $\x(k) = \cR_i  ( x_{k,i} )$,  $(k,i) \in \cE$. Thus, $\y(k) = ( \cQ_i \circ \cR_i )  ( x_{k,i} )$, which shows that $(\bY,\cT)$ is a solution to \ref{reg}, where $\bY = ( \y(k) )_{k=1}^N$ and $\cT = ( \cQ_i \circ \cR_i )_{i=1}^M$.	
\end{proof}

\noindent Foregoing definitions and propositions allow us to express the condition of unique registrability in a compact manner. Namely, let $(\Gamma_C,\bar{\x})$ be the ground-truth node-patch framework. Then, under assumption (A1), \ref{reg} has a unique solution if and only if for any node-patch framework $(\Gamma_C,\y)$ such that $(\Gamma_C,\y) \sim (\Gamma_C,\bar{\x})$, we have $(\Gamma_C,\y) \equiv (\Gamma_C,\bar{\x})$.
The next two propositions relate node-patch framework and body graph framework.

\begin{proposition} \label{prop:equiveq}
	Two node-patch frameworks $(\Gamma_C,\textup{\x})$ and $(\Gamma_C,\y)$ are equivalent (Def. \ref{def:ppequiv}) if and only if the body graph frameworks $(\Gamma_B,\textup{\x})$ and $(\Gamma_B,\y)$ are equivalent (Def. \ref{def:eqframe}).
\end{proposition}

\begin{proof}
	Suppose $(\Gamma_C,\x) \sim (\Gamma_C,\y)$. Pick an arbitrary edge $(k,l) \in E$ in the body graph $\Gamma_B = (V,E)$. From construction of $\Gamma_B$, $(k,l) \in E$ if and only if there is a patch, say $P_i$, that contains both the nodes $k$ and $l$. Since $(\Gamma_C,\x) \sim (\Gamma_C,\y)$, there exists a rigid transform $\cQ_i$ such that $\x(k) = \cQ_i  \y(k)$ and $\x(l) = \cQ_i  \y(l)$. This implies that $\x(k) - \x(l) = \cQ_i  ( \y(k) - \y(l) )$, from where it follows that $\norm{\x(k) - \x(l)} = \norm{\y(k) - \y(l) )}$. Thus, $(\Gamma_B,\x) \sim (\Gamma_B,\y)$.
	
	Conversely, suppose $(\Gamma_B,\x) \sim (\Gamma_B,\y)$. Consider an arbitrary patch $P_i$. Note that any subgraph of $\Gamma_B$ induced by a patch is a clique. This, along with the assumption that $(\Gamma_B,\x) \sim (\Gamma_B,\y)$, implies that $\norm{\x(k) - \x(l)} = \norm{\y(k) - \y(l) )}$ for every $k,l \in P_i$, which, in turn, implies that there exists a rigid transform $\cQ_i$ such that $\x(v) = \cQ_i  \y(v)$, $v \in P_i$. Thus, $(\Gamma_C,\x) \sim (\Gamma_C,\y)$. 
\end{proof}

\begin{proposition} \label{prop:congeq}
	Two node-patch frameworks $(\Gamma_C,\textup{\x})$ and $(\Gamma_C,\y)$ are congruent (Def. \ref{def:ppcong}) if and only if the body graph frameworks $(\Gamma_B,\textup{\x})$ and $(\Gamma_B,\y)$ are congruent (Def. \ref{def:congframe}).
\end{proposition}

The above result easily follows from Definitions \ref{def:congframe} and \ref{def:ppcong}. We can now complete the proof of Theorem \ref{thm:mainthm}. Suppose \ref{reg} has a unique solution. We will show that the body graph framework $(\Gamma_B,\bar{\x})$ is globally rigid. Consider a framework $(\Gamma_B,\y) \sim (\Gamma_B,\bar{\x})$. Then, by Proposition \ref{prop:equiveq}, $(\Gamma_C,\y) \sim (\Gamma_C,\bar{\x})$. By Proposition \ref{prop:equivsol}, this implies that $(\Gamma_C,\y)$ correponds to a solution of \ref{reg}. Now, since \ref{reg} has a unique solution, $(\Gamma_C,\y) \equiv (\Gamma_C,\bar{\x})$. Thus, by Proposition \ref{prop:congeq}, $(\Gamma_B,\y) \equiv (\Gamma_B,\bar{\x})$.

Conversely, suppose $(\Gamma_B,\bar{\x})$ is globally rigid. Let $(\bY,\cT)$ be a solution to \ref{reg}. By Proposition \ref{prop:solequiv}, $(\Gamma_C,\y) \sim (\Gamma_C,\bar{\x})$. Hence, by Proposition \ref{prop:equiveq}, $(\Gamma_B,\y) \sim (\Gamma_B,\bar{\x})$. This, by global rigidity of $(\Gamma_B,\bar{\x})$, implies that $(\Gamma_B,\y) \equiv (\Gamma_B,\bar{\x})$. Finally, by Proposition \ref{prop:congeq}, $(\Gamma_C,\y) \equiv (\Gamma_C,\bar{\x})$.

\subsection{Proof of Theorem \ref{thm:body3suff}.} \label{subsec:body3suffproof}

To prove Theorem \ref{thm:body3suff}, we need the following proposition (similar observation was made in \cite{jordan2016}).

\begin{proposition} \label{prop:addvertexthm}
	Given a graph $G = (V,E)$, consider the graph $G' = (V\cup\{v'\},E')$ obtained by adding a new vertex $v'$ to $G$ and attaching it to a clique $H \subseteq G$, i.e., $v'$ is adjacent to every vertex of $H$ and to no other vertex of $G$. If $G'$ is generically globally rigid, then $G$ is generically globally rigid.
\end{proposition}

\begin{proof}
	Suppose $G$ is not generically globally rigid. Consider two frameworks $(G,\mathbf{p})$ and $(G,\mathbf{q})$ which are equivalent but not congruent. To these frameworks, add the new vertex $v'$ to get new frameworks $(G',\mathbf{p'})$ and $(G',\mathbf{q'})$ such that the distance between $v'$ and any vertex of the subgraph $H$ is equal in both $(G',\mathbf{p'})$ and $(G',\mathbf{q'})$. Note that this can be done because $H$ is a clique and so the subframeworks induced by $H$ would be congruent in the two frameworks $(G,\mathbf{p})$ and $(G,\mathbf{q})$. Clearly, the new frameworks $(G',\mathbf{p'})$ and $(G',\mathbf{q'})$ are equivalent. But they are not congruent because $(G,\mathbf{p})$ and $(G,\mathbf{q})$ were not congruent to begin with. Thus, $G'$ is not generically globally rigid. 	
\end{proof}

We now prove Theorem \ref{thm:body3suff}. The necessity of $3$-connectivity of the body graph $\Gamma_B$ for unique registrability in $\R^2$ follows from Theorem \ref{thm:mainthm} and Theorem \ref{thm:hendrick}. We now establish sufficiency.
Given that the body graph $\Gamma_B$ is $3$-connected, we will prove that $\Gamma_B$ is generically globally rigid in $\R^2$; this, by Theorem \ref{thm:mainthm}, would imply unique registrability in $\R^2$. By Assumption (A1), there are at least $3$ nodes in each patch. Consider the following cases:
\begin{enumerate}[]
		
	\item \textbf{Case 1: Each patch contains at least $4$ nodes.}		
	Pick an arbitrary edge $(k,l)$ belonging to $\Gamma_B$. The fact that there is an edge between vertices $k$ and $l$ implies that there must be a patch, say $P_i$, which contains the nodes $k$ and $l$. Since $P_i$ contains at least $4$ nodes, we can pick two nodes $\bar{k}$ and $\bar{l}$ belonging to $P_i$ which are distinct from the nodes $k$ and $l$. Now, $P_i$ induces a clique, say $H_i$, in $\Gamma_B$. This implies that the subgraph of $\Gamma_B$ induced by the vertex set $\{k,l,\bar{k},\bar{l}\}$ is $K_4$, which, in particular, means that the edge $(k,l)$ belongs to $K_4$. The edge $(k,l)$ was chosen arbitrarily, and thus, we have shown that every edge of $\Gamma_B$ belongs to $K_4$. Since $\Gamma_B$ is also $3$-connected, Proposition \ref{prop:k4redundant} leads us to conclude that $\Gamma_B$ is redundantly rigid. Thus, $\Gamma_B$ satisfies conditions in Theorem \ref{thm:jj2d}, and is hence generically globally rigid in $\R^2$.
		
	\item \textbf{Case 2: There are patches with exactly $3$ nodes.}
	Suppose there are $m$ patches $P_1,\cdots,P_m$ that contain exactly $3$ nodes. Add a new node $k_1$ exclusively to patch $P_1$ and call the resulting patch $P'_1$. The effect of this on the body graph is the addition of a degree-$3$ vertex $k_1$ adjacent to the vertices of the clique induced by the $3$ nodes in $P_1$. Call the resulting body graph $\Gamma^1_B$. Addition of a degree-$k$ vertex to a $k$-connected graph results in a $k$-connected graph. Thus, $\Gamma^1_B$ is $3$-connected. We continue inductively: after obtaining $\Gamma^i_B$, add a new node $k_{i+1}$ exclusively to patch $P_{i+1}$ to get $P'_{i+1}$ and the resulting body graph $\Gamma^{i+1}_B$. Note that we preserve $3$-connectivity at every step of the induction. We stop after we have obtained the body graph $\Gamma^m_B$. As a result of this inductive procedure, every patch now contains at least $4$ nodes. Hence, from the arguments made in \emph{Case 1} above, $\Gamma^m_B$ is generically globally rigid in $\R^2$. Now, $\Gamma^m_B$ was obtained from $\Gamma^{m-1}_B$ by addition of a vertex and attaching it to a clique. Hence, from Proposition \ref{prop:addvertexthm}, $\Gamma^{m-1}_B$ is generically globally rigid in $\R^2$. Backtracking similarly in an inductive fashion and employing Proposition \ref{prop:addvertexthm} at every step, we deduce that the original body graph $\Gamma_B$ is generically globally rigid in $\R^2$.
\end{enumerate}


\subsection{Proof of Theorem \ref{thm:quasibodykconn}.} \label{subsec:proofquasibodykconn}
	We first prove Theorem \ref{thm:quasibodykconn}.$(ii)$. We are given that every patch has at least $k$ nodes and the body graph $\Gamma_B$ is $k$-connected. Let $H_i$ and $H_j$ be the cliques of $\Gamma_B$ induced by patches $P_i$ and $P_j$, $i \neq j$. To establish quasi $k$-connectivity of $\Gamma_C$, it suffices to show that there exists $k$ disjoint $H_i$-$H_j$ paths. Indeed, it is clear from Definition \ref{quasiconndef} that the existence of $k$ disjoint $H_i$-$H_j$ paths in $\Gamma_B$ implies the existence of $k$ $\cS$-disjoint paths in $\Gamma_C$ between $P_i$ and $P_j$. Add two new vertices $a$ and $b$ to $\Gamma_B$ such that $a$ is adjacent to every vertex of  $H_i$ (and to no other vertex of $\Gamma_B$), and $b$ is adjacent to every vertex of  $H_j$ (and to no other vertex of $\Gamma_B$). Since each patch has at least $k$ nodes, $\mathrm{degree}(a) \geq k$ and $\mathrm{degree}(b) \geq k$. Addition of a degree-$k$ vertex to a $k$-connected graph results in a $k$-connected graph. Thus, the graph obtained after adding $a$ and $b$ to $\Gamma_B$ is $k$-connected. This implies that there are at least $k$ independent paths between $a$ and $b$. Now, each such path has to be of the form $a-v_1-\cdots-v_r-b$, where $v_1 \in H_i$ and $v_r \in H_j$. This is because $a$ is adjacent only to vertices from $H_i$ and $b$ is adjacent only to vertices from $H_j$. Removing $a$ and $b$ from every such independent path gives us $k$ disjoint $H_i$-$H_j$ paths. 
	
	We now prove Theorem \ref{thm:quasibodykconn}.$(i)$. Assume, without loss of generality, that no two patches are identical. To prove $k$-connectivity of the body graph $\Gamma_B = (V,E)$, we will show that given arbitrary vertices $a,b \in V$, there exists $k$ independent paths between them. We consider the following cases:
	\begin{enumerate}[]
		\item \textbf{Case 1: $a$ and $b$ do not belong to the same patch.}
		Suppose $a \in P_i$ and $b \in P_j$, where $i \neq j$. Denote the cliques of $\Gamma_B$ induced by patches $P_i$ and $P_j$ as $H_i$ and $H_j$. Since $\Gamma_C$ is quasi $k$-connected, there exists $k$ disjoint $H_i$-$H_j$ paths (Observation \ref{obs:defquasi}). Note that a vertex in $V(H_i) \cap V(H_j)$ is also considered an $H_i$-$H_j$ path. Let $P = v_1-\cdots-v_r$ be one such path, where $v_1 \in H_i$ and $v_r \in H_j$. Since $H_i$ and $H_j$ are cliques, $(a,v_1) \in E$ and $(v_r,b) \in E$. Thus for each of the $k$ disjoint $H_i$-$H_j$ paths, we can, if needed, append vertices $a$ and $b$ at the ends to make it of the form $a-\cdots-b$. For instance, if $v_1 \neq a$ and $v_r \neq b$, we modify the path to $a-v_1-\cdots-v_r-b$. Thus, we have $k$ independent paths between $a$ and $b$.
		
		\item \textbf{Case 2: $a$ and $b$ belong to the same patch.}
		Suppose $a$ and $b$ belong to patch $P_l$. Quasi $k$-connectivity of the correspondence graph implies that each patch has at least $k$ participating nodes (Observation \ref{obs:defquasi}).  In particular, this means that the clique $H_l$ of $\Gamma_B$ induced by $P_l$ has at least $k$ vertices. Thus, if $a$ and $b$ belong to $P_l$, there are at least $k-1$ independent paths within the clique $H_l$. If $P_l$ has more than $k$ nodes, we thus get $k$ independent paths between $a$ and $b$, all from within $H_l$. But suppose $P_l$ has exactly $k$ nodes. We need an additional path between $a$ and $b$ that is independent of the $k-1$ paths we have from within $H_l$. Since we have exactly $k$ nodes in $P_l$, each node has to be participating, i.e., each node belongs to at least $2$ patches. We consider the following sub-cases:
		
		\begin{enumerate}[]
			\item \textbf{Sub-case I: There is a patch $P_i$, $i \neq l$, containing both $a$ and $b$.}
			In this case we get the additional path of the form $a-v-b$, where $v \in P_i$ and $v \notin P_l$, which, clearly, is independent of the $k-1$ paths from within $H_l$. The assumption that no two patches are identical ensures the existence of the $v$ in question.
			
			\item \textbf{Sub-case II: There is no patch other than $P_l$ containing both $a$ and $b$.}
			Suppose $a \in P_i$ and $b \in P_j$, $i \neq j$. From the quasi $k$-connectivity assumption, we know there are $k$ disjoint $H_i$-$H_j$ paths. Moreover, recall that there are exactly $k$ vertices in $H_l$. Consider the following possibilities:
			
			\begin{enumerate}[(i)]
				\item Suppose every disjoint $H_i$-$H_j$ path contains a vertex from $H_l$. This is possible if and only if each path contains exactly one vertex from $H_l$. In this case, there exists a path of the form $a-v_1-\cdots-v_r$, such that $v_1,\cdots,v_r \notin H_l$, and $v_r \in H_j$. From completeness of the clique $H_j$, we can append $b$ to the end of this path to get $a-v_1-\cdots-v_r-b$. This path is independent of the $k-1$ paths we have from within $H_l$. Thus we have the required additional path.
				
				\item The only other case is when there exists a disjoint $H_i$-$H_j$ path that has no vertex from $H_l$. Let that path be $v_1-\cdots-v_r$ where $v_1 \in H_i$ and $v_r \in H_j$. From completeness of the cliques $H_i$ and $H_j$, we can append $a$ and $b$ to the ends of this path to get $a-v_1-\cdots-v_r-b$, which is independent of the $k-1$ paths we have from within $H_l$. Again, we have the required additional path.
			\end{enumerate}
		\end{enumerate}			
	\end{enumerate}

\bibliographystyle{IEEEtran}
\bibliography{bibrefIEEE}

%

\end{document}